\documentclass[letterpaper, 10 pt, conference]{ieeeconf}  

\IEEEoverridecommandlockouts                              

\usepackage{amsmath,amssymb,amsfonts}
\usepackage{lipsum}
\usepackage{mathtools}
\usepackage{cuted}
\usepackage{cite}

\newtheorem{theorem}{\it Theorem}
\newtheorem{lemma}{\it Lemma}
\newtheorem{definition}{\it Definition}

\newtheorem{corollary}{\it Corollary}
\newtheorem{proposition}{\it Proposition}

\newtheorem{scenario}{\it Scenario}

\usepackage{balance}

\DeclareMathOperator*{\esssup}{ess\,sup}

\title{\LARGE \bf
Information-Theoretic Performance Limitations of Feedback Control: Underlying Entropic Laws and Generic $\mathcal{L}_{p}$ Bounds
}

\author{Song Fang$^{1}$ and Quanyan Zhu$^{1}$
	\thanks{$^{1}$Song Fang and Quanyan Zhu are with the Department of Electrical and Computer Engineering, New York University, New York, USA
		{\tt\small song.fang@nyu.edu; quanyan.zhu@nyu.edu}}%
}

\begin{document}

\maketitle
\thispagestyle{empty}
\pagestyle{empty}

\begin{abstract}

In this paper, we utilize information theory to study the fundamental performance limitations of generic feedback systems, where both the controller and the plant may be any causal functions/mappings while the disturbance can be with any distributions. More specifically, we obtain fundamental $\mathcal{L}_p$ bounds on the control error, which are shown to be completely characterized by the conditional entropy of the disturbance, based upon the entropic laws that are inherent in any feedback systems. We also discuss the generality and implications (in, e.g., fundamental limits of learning-based control) of the obtained bounds.

\end{abstract}


\section{Introduction}
\label{sec:intro}


Machine learning techniques  are becoming more and more prevalent nowadays in the feedback control of dynamical systems (see, e.g., \cite{
	duriez2017machine,
	kiumarsi2017optimal, bertsekas2019reinforcement,recht2019tour,
	zoppoli2020neural} and the references therein), where system dynamics that are determined by physical laws will play an indispensable role. Representative learning-based control systems include control with deep reinforcement learning \cite{mnih2015human, duan2016benchmarking}, where, as it is name indicates, deep reinforcement learning algorithms are employed to replace conventional controllers in feedback control systems such as robotic control systems. In this trend, it is becoming more and more critical to be fully aware of the performance limits of the machine learning algorithms that are to be embedded in the feedback loop, i.e., the fundamental limits of learning-based control, especially in scenarios where performance guarantees are required and must be strictly imposed. This is equivalent to ask: What is the optimal control performance that machine learning methods can achieve in the position of the controller?

In conventional performance limitation analysis \cite{seron2012fundamental} of feedback systems such as the Bode integral \cite{Bod:45}, however, specific restrictions on the classes of the controller that can be implemented must be imposed in general; one common restriction is that the controllers are assumed to be linear time-invariant (LTI) \cite{seron2012fundamental}. These restrictions would normally render the analysis invalid if machine learning elements such as deep learning \cite{goodfellow2016deep} or reinforcement learning \cite{bertsekas2019reinforcement} are to be placed at the position of the controller, since the learning algorithms are quite complicated and not necessarily LTI from an input-output viewpoint.



Information theory \cite{Cov:06}, a mathematical theory developed originally for the analysis of fundamental limits of communication systems \cite{shannon1998mathematical}, was in recent years seen to be applicable to the analysis of performance limitations of feedback control systems as well, including Bode-type integrals \cite{zang2003nonlinear, Mar:07, Mar:08, Oka:08, Ish:09, Yu:10, Les:10,  hurtado2010limitations, Li:13a, Li:13b, heertjes2013self, ruan2013information, Zha:14, zhao2015effect,  lupu2015information, Fang17TAC, Fang17Automatica, wan2019sensitivity} (see also \cite{fang2017towards, chen2019fundamental} for surveys on this topic), power gain bounds \cite{fang2018power}, and limits of variance minimization \cite{fang2017fundamental}.
One essential difference between this line of research and the conventional feedback performance limitation analysis is that the information-theoretic performance bounds hold for any causal controllers (of which LTI controllers are only a special class), i.e., for any causal functions/mappings in the position of the controller, though oftentimes the plant is assumed to be LTI while the disturbance is assumed stationary Gaussian. This is a key observation that enables the later discussions on the implications of such results in fundamental limits of learning-based control. (It is worth mentioning that there also exist other approaches to analyze the performance limits of feedback control systems while allowing the controllers to be generic; see, e.g., \cite{xie2000much, guo2020feedback} and \cite{nakahira2019connecting,nakahira2020integrative} as well as the references therein.)

In this paper, we go beyond the classes of performance limitations analyzed in the aforementioned works, and investigate the fundamental $\mathcal{L}_p$ limitations of generic feedback systems in which both the controller and the plant may be any causal while the disturbance can be with any distributions and is not necessarily stationary. The analysis will be carried out by studying the underlying entropic relationships of the signals flowing the feedback loop, and the derived $\mathcal{L}_p$ bounds (for $p \geq 1$) are all seen to be completely characterized by the conditional entropy of the disturbance.
We also discuss the implications of the derived bounds in fundamental limits of learning-based control, noticing that any machine learning elements in the position of the controller may be viewed as causal functions/mappings from its input to output.  

The remainder of the paper is organized as follows. Section~II introduces the technical preliminaries. In Section~III, we introduce the fundamental $\mathcal{L}_{p}$ bounds for generic feedback systems, with discussions on their generality and implications.
Concluding remarks are given in Section~IV.

Note that an arXiv version of this paper with additional results and discussions can be found in \cite{FangCSL20arxiv}.



\section{Preliminaries}

In this paper, we consider real-valued continuous random variables and vectors, as well as discrete-time stochastic processes they compose. All the random variables, random vectors, and stochastic processes are assumed to be zero-mean, for simplicity and without loss of generality. We represent random variables and vectors using boldface letters. Given a stochastic process $\left\{ \mathbf{x}_{k}\right\}$, we denote the sequence $\mathbf{x}_0,\ldots,\mathbf{x}_{k}$ by the random vector $\mathbf{x}_{0,\ldots,k}=\left[\mathbf{x}_0^T~\cdots~\mathbf{x}_{k}^T\right]^T$ for simplicity. The logarithm is defined with base $2$. All functions are assumed to be measurable. 
A stochastic process $\left\{ \mathbf{x}_{k}\right\}$ is said to be asymptotically stationary if it is stationary as $k \to \infty$, and herein stationarity means strict stationarity unless otherwise specified \cite{Pap:02}. 
In addition, a process being asymptotically stationary implies that it is asymptotically mean stationary \cite{gray2011entropy}. 
Note in particular that, for simplicity and
with abuse of notations, we utilize $\mathbf{x} \in \mathbb{R}$ and $\mathbf{x} \in \mathbb{R}^m$ to
indicate that $\mathbf{x}$ is a real-valued random variable and that $\mathbf{x}$
is a real-valued $m$-dimensional random vector, respectively.

Entropy and mutual information are the most basic notions in information theory \cite{Cov:06}, which we introduce below.

\begin{definition} The differential entropy of a random vector $\mathbf{x}$ with density $p_{\mathbf{x}} \left(x\right)$ is defined as
	\begin{flalign}
	h\left( \mathbf{x} \right)
	=-\int p_{\mathbf{x}} \left(x\right) \log p_{\mathbf{x}} \left(x\right) \mathrm{d} x. \nonumber
	\end{flalign}
	The conditional differential entropy of random vector $\mathbf{x}$ given random vector $\mathbf{y}$ with joint density $p_{\mathbf{x}, \mathbf{y}} \left(x,y\right)$ and conditional density $p_{\mathbf{x} | \mathbf{y}} \left(x,y\right)$ is defined as
	\begin{flalign}
	h\left(\mathbf{x}\middle|\mathbf{y}\right)
	=-\int p_{\mathbf{x}, \mathbf{y}} \left(x,y\right)\log p_{\mathbf{x} | \mathbf{y}} \left(x,y\right) \mathrm{d}x\mathrm{d}y. \nonumber
	\end{flalign}
	The mutual information between random vectors $\mathbf{x}, \mathbf{y}$ with densities $p_{\mathbf{x}} \left(x\right)$, $p_{\mathbf{y}} \left( y \right) $ and joint density $p_{\mathbf{x}, \mathbf{y}} \left(x,y\right)$ is defined as
	\begin{flalign}
	I\left(\mathbf{x};\mathbf{y}\right)
	=\int p_{\mathbf{x}, \mathbf{y}} \left(x,y\right) \log \frac{p_{\mathbf{x}, \mathbf{y}} \left(x,y\right)}{p_{\mathbf{x}} \left(x\right) p_{\mathbf{y}} \left( y \right) }\mathrm{d}x\mathrm{d}y. \nonumber
	\end{flalign}
	The entropy rate of a stochastic process $\left\{ \mathbf{x}_{k}\right\}$ is defined as
	\begin{flalign}
	h_\infty \left(\mathbf{x}\right)=\limsup_{k\to \infty} \frac{h\left(\mathbf{x}_{0,\ldots,k}\right)}{k+1}. \nonumber
	\end{flalign}
\end{definition}

\vspace{3mm}

Properties of these notions can be found in, e.g., \cite{Cov:06}. 
In particular, the next lemma \cite{dolinar1991maximum} presents the maximum-entropy probability distributions under $\mathcal{L}_{p}$-norm constraints for random variables.

\begin{lemma} \label{maximum}
	Consider a random variable $\mathbf{x} \in \mathbb{R}$ with $\mathcal{L}_{p}$ norm $\left[ \mathbb{E} \left( \left| \mathbf{x} \right|^{p} \right) \right]^{\frac{1}{p}} = \mu,~p \geq 1$.
	Then,  
	\begin{flalign} 
	h \left( \mathbf{x} \right) 
	\leq \log \left[ 2 \Gamma \left( \frac{p+1}{p} \right) \left( p \mathrm{e} \right)^{\frac{1}{p}} \mu \right], \nonumber
	\end{flalign}
	where equality holds if and only if $\mathbf{x}$ is with probability density
	\begin{flalign}
	f_{\mathbf{x}} \left( x \right)
	= \frac{ \mathrm{e}^{- \left| x \right|^{p} / \left( p \mu^{p} \right)} }{2 \Gamma \left( \frac{p+1}{p} \right) p^{\frac{1}{p}} \mu}. \nonumber
	\end{flalign}
	Herein, $\Gamma \left( \cdot \right)$ denotes the Gamma function.
\end{lemma}

In particular, when $p \to \infty$, 
\begin{flalign}
\lim_{p \to \infty} \left[ \mathbb{E} \left( \left| \mathbf{x} \right|^{p} \right) \right]^{\frac{1}{p}} = \esssup_{ f_{\mathbf{x}} \left( x \right) > 0} \left| \mathbf{x} \right|, \nonumber
\end{flalign}
and
\begin{flalign} 
\lim_{p \to \infty} \log \left[ 2 \Gamma \left( \frac{p+1}{p} \right) \left( p \mathrm{e} \right)^{\frac{1}{p}} \mu \right] = \log \left( 2 \mu \right), \nonumber
\end{flalign}
while
\begin{flalign}
\lim_{p \to \infty}\frac{ \mathrm{e}^{- \left| x \right|^{p} / \left( p \mu^{p} \right)} }{2 \Gamma \left( \frac{p+1}{p} \right) p^{\frac{1}{p}} \mu}
= 
\left\{ \begin{array}{cc}
\frac{1}{2 \mu}, & \left| x \right| \leq \mu,\\
0, & \left| x \right| > \mu.
\end{array} \right. \nonumber
\end{flalign}

\section{Fundamental $\mathcal{L}_{p}$ Bounds of Feedback Systems}

\begin{figure}
	\begin{center}
		\vspace{-6mm}
		\includegraphics [width=0.4\textwidth]{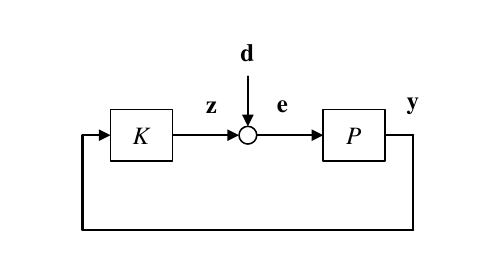}
		\vspace{-6mm}
		\caption{A feedback control system.}
		\label{feedback}
	\end{center}
	\vspace{-3mm}
\end{figure}

In this section, we present fundamental $\mathcal{L}_{p}$ bounds for generic feedback systems. We start with two scenarios.

\begin{scenario}
	Consider first the feedback control system depicted in Fig.~\ref{feedback}. Herein, $\mathbf{d}_{k}, \mathbf{e}_{k}, \mathbf{z}_{k}, \mathbf{y}_{k} \in \mathbb{R}$. Assume that the plant $P$ is strictly causal, i.e., 
\begin{flalign} 
\mathbf{y}_{k}= {P}_{k} \left( \mathbf{e}_{0,\ldots,k-1}\right), \nonumber 
\end{flalign}
for any time instant $k \geq 0$. Moreover, the controller $ {K}$ is assumed to be causal, i.e., for any time instant $k \geq 0$,
\begin{flalign}
\mathbf{z}_{k}= {K}_{k} \left( \mathbf{y}_{0,\ldots,k}\right). \nonumber
\end{flalign}
In fact, both ${K}_{k} \left( \cdot \right)$ and ${P}_{k} \left( \cdot \right)$ can be deterministic functions/mappings as well as randomized functions/mappings.
Furthermore, the disturbance $\left\{ \mathbf{d}_{k} \right\}$ and the initial state $\mathbf{z}_0$ are assumed to be independent. Note that if
the plant is not strictly causal, then the controller
should be assumed strictly causal so as to ensure the strict
causality of the open-loop system, thus preventing $\left\{ \mathbf{d}_{k} \right\}$ and $\mathbf{z}_0$ from being dependent while also avoiding any other causality issues that might arise in feedback systems.
\end{scenario}

\begin{figure}
	\begin{center}
		\vspace{-6mm}
		\includegraphics [width=0.4\textwidth]{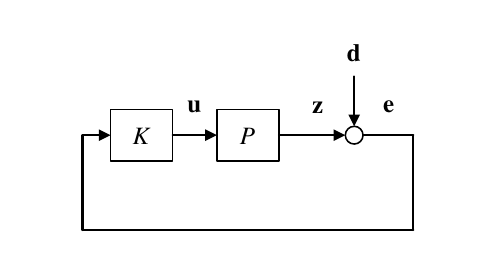}
		\vspace{-6mm}
		\caption{Another feedback control system.}
		\label{feedback2}
	\end{center}
	\vspace{-6mm}
\end{figure}

\begin{scenario}
We may also consider another system setting as depicted in Fig.~\ref{feedback2}. Herein, $\mathbf{d}_{k}, \mathbf{e}_{k}, \mathbf{z}_{k}, \mathbf{u}_{k} \in \mathbb{R}$. Assume that the plant $P$ is strictly causal, i.e., 
\begin{flalign} 
\mathbf{z}_{k}= {P}_{k} \left( \mathbf{u}_{0,\ldots,k-1}\right). \nonumber 
\end{flalign}
Moreover, the controller $ {K}$ is assumed to be causal, i.e.,
\begin{flalign}
\mathbf{u}_{k}= {K}_{k} \left( \mathbf{e}_{0,\ldots,k}\right). \nonumber
\end{flalign}
Furthermore, $\left\{ \mathbf{d}_{k} \right\}$ and $\mathbf{z}_0$ are assumed to be independent. Again, if
the plant is not strictly causal, then the controller
should be assumed strictly causal to ensure the strict
causality of the open-loop system. 
\end{scenario}

\begin{figure}
	\begin{center}
		\vspace{-6mm}
		\includegraphics [width=0.3\textwidth]{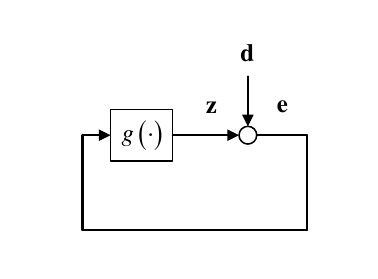}
		\vspace{-6mm}
		\caption{A generic feedback system.}
		\label{feedback3}
	\end{center}
	\vspace{-6mm}
\end{figure}

In fact, both systems of Scenario~1 (Fig.~\ref{feedback}) and Scenario~2 (Fig.~\ref{feedback2}) can be viewed as special cases of the unifying framework depicted in Fig.~\ref{feedback3}, as a generic feedback system. More specifically, herein $\mathbf{d}_{k}, \mathbf{e}_{k}, \mathbf{z}_{k} \in \mathbb{R}$, and ${g} \left( \cdot \right)$ is assumed to be strictly causal, i.e., 
\begin{flalign}  \label{generic}
\mathbf{z}_{k}= {g}_{k} \left( \mathbf{e}_{0,\ldots,k-1}\right), 
\end{flalign}
for any time instant $k \geq 0$. 
Furthermore, $\left\{ \mathbf{d}_{k} \right\}$ and $\mathbf{z}_0$ are assumed to be independent. As such, the system in Fig.~\ref{feedback} may be viewed the a special case of that in Fig.~\ref{feedback3} for the case of $g \left( \cdot \right) = K \left( P \left( \cdot \right) \right)$, while that in Fig.~\ref{feedback2} is a special case for when $g \left( \cdot \right) = P \left( K \left( \cdot \right) \right)$. Accordingly, to investigate the disturbance attenuation properties (essentially the relationship between the error signal $\left\{ \mathbf{e}_{k} \right\}$ and the disturbance $\left\{ \mathbf{d}_{k} \right\}$) of the systems in Fig.~\ref{feedback} or Fig.~\ref{feedback2}, it suffices to examine the generic feedback system given in Fig.~\ref{feedback3}.

As the main result of this paper, we now show that the following $\mathcal{L}_{p}$ bound on the error signal $\left\{ \mathbf{e}_{k} \right\}$ always holds for generic feedback systems.

\begin{theorem} \label{recursivetheorem}
	Consider the system given in Fig.~\ref{feedback3}.
	Then,  
	\begin{flalign} \label{feedbackbound}
	\left[ \mathbb{E} \left( \left| \mathbf{e}_{k} \right|^{p} \right) \right]^{\frac{1}{p}}
	\geq \frac{2^{h \left( \mathbf{d}_k | \mathbf{d}_{0,\ldots,k-1} \right)}}{2 \Gamma \left( \frac{p+1}{p} \right) \left( p \mathrm{e} \right)^{\frac{1}{p}}},
	\end{flalign}
	where equality holds if and only if $\mathbf{e}_{k}$ is with probability density 
	\begin{flalign} \label{distribution}
	f_{\mathbf{e}_{k}} \left( x \right)
	= \frac{ \mathrm{e}^{- \left| x \right|^{p} / \left( p \mu^{p} \right)} }{2 \Gamma \left( \frac{p+1}{p} \right) p^{\frac{1}{p}} \mu},
	\end{flalign}
	and $I \left( \mathbf{e}_{k}; \mathbf{d}_{0,\ldots,k-1}, \mathbf{z}_{0} \right) = 0$.
\end{theorem}

\begin{proof}
To begin with, it follows from Lemma~\ref{maximum} that
	\begin{flalign} 
	\left[ \mathbb{E} \left( \left| \mathbf{e}_{k} \right|^{p} \right) \right]^{\frac{1}{p}}
	\geq \frac{2^{h \left(  \mathbf{e}_{k} \right)}}{2 \Gamma \left( \frac{p+1}{p} \right) \left( p \mathrm{e} \right)^{\frac{1}{p}}}, \nonumber
	\end{flalign}
	where equality holds if and only if $\mathbf{e}_{k}$ is with probability density
	\begin{flalign}
	f_{\mathbf{e}_{k}} \left( x \right)
	= \frac{ \mathrm{e}^{- \left| x \right|^{p} / \left( p \mu^{p} \right)} }{2 \Gamma \left( \frac{p+1}{p} \right) p^{\frac{1}{p}} \mu}. \nonumber
	\end{flalign}	
	Herein, $\mu$ is a normalizing factor; in fact, when equality is achieved herein, it holds that
\begin{flalign}
\mu 
= \frac{2^{h \left( \mathbf{e}_{k} \right)}}{2 \Gamma \left( \frac{p+1}{p} \right) \left( p \mathrm{e} \right)^{\frac{1}{p}}}. \nonumber
\end{flalign}

On the other hand, we will prove the fact that $\mathbf{z}_{k}$ is eventually a function of $\mathbf{d}_{0,\ldots,k-1}$ and $\mathbf{z}_{0}$.
More specifically, it is clear that when $k=0$, \eqref{generic} reduces to  
$$
\mathbf{z}_{1} = g_{0} \left( \mathbf{e}_{0} \right) = g_{0} \left( \mathbf{d}_{0} + \mathbf{z}_{0} \right),
$$
that is, $\mathbf{z}_{1}$ is a function of $\mathbf{d}_{0}$ and $\mathbf{z}_{0}$.
Next, when $k=1$, 
\eqref{generic} is given by 
$$
\mathbf{z}_{2} = g_{1} \left( \mathbf{e}_{0}, \mathbf{e}_{1} \right) = g_{1} \left( \mathbf{d}_{0} +  \mathbf{z}_{0}, \mathbf{d}_{1} +  \mathbf{z}_{1} \right). 
$$
As such, since $\mathbf{z}_{1}$ is a function of $\mathbf{d}_{0}$ and $\mathbf{z}_{0}$, we have 
$$
\mathbf{z}_{2} = g_{1} \left( \mathbf{d}_{0} +  \mathbf{z}_{0}, \mathbf{d}_{1} +  g_{0} \left( \mathbf{d}_{0} + \mathbf{z}_{0} \right) \right). 
$$
In other words, $\mathbf{z}_{2}$ is a function of $\mathbf{d}_{0,1}$ and $\mathbf{z}_{0}$. We may then repeat this process and show that for any $k \geq 0$, $\mathbf{z}_{k}$ is eventually a function of $\mathbf{d}_{0,\ldots,k-1}$ and $\mathbf{z}_{0}$.

We will then proceed to prove the main result of this theorem. Note first that
\begin{flalign} 
h \left( \mathbf{e}_{k} \right)
& = h \left( \mathbf{e}_{k} |  \mathbf{d}_{0,\ldots,k-1}, \mathbf{z}_{0} \right) + I \left( \mathbf{e}_{k}; \mathbf{d}_{0,\ldots,k-1}, \mathbf{z}_{0} \right) \nonumber \\
& = h \left( \mathbf{z}_{k} + \mathbf{d}_{k} |  \mathbf{d}_{0,\ldots,k-1}, \mathbf{z}_{0} \right) + I \left( \mathbf{e}_{k}; \mathbf{d}_{0,\ldots,k-1}, \mathbf{z}_{0} \right). \nonumber
\end{flalign}
Then, according to the fact that $\mathbf{z}_{k}$ is a function of $\mathbf{d}_{0,\ldots,k-1}$ and $\mathbf{z}_{0}$, we have
\begin{flalign} 
h \left( \mathbf{z}_{k} + \mathbf{d}_{k} |  \mathbf{d}_{0,\ldots,k-1}, \mathbf{z}_{0} \right) = h \left( \mathbf{d}_{k} |  \mathbf{d}_{0,\ldots,k-1}, \mathbf{z}_{0} \right). \nonumber
\end{flalign}
On the other hand, since $\mathbf{z}_{0}$ and $ \left\{ \mathbf{d}_{k} \right\}$ are independent (and thus $\mathbf{z}_{0}$ and $\mathbf{d}_{k}$ are independent given $\mathbf{d}_{0,\ldots,k-1}$), we have
\begin{flalign}  
& h \left( \mathbf{d}_{k} |  \mathbf{d}_{0,\ldots,k-1}, \mathbf{z}_{0} \right) \nonumber \\
&~~~~ = h \left( \mathbf{d}_{k} |  \mathbf{d}_{0,\ldots,k-1} \right)  - I \left( \mathbf{d}_{k}; \mathbf{z}_{0} |  \mathbf{d}_{0,\ldots,k-1} \right) \nonumber \\
&~~~~ = h \left( \mathbf{d}_{k} |  \mathbf{d}_{0,\ldots,k-1} \right). \nonumber
\end{flalign}
As a result,
\begin{flalign} 
h \left( \mathbf{e}_{k} \right)
= h \left( \mathbf{d}_{k} |  \mathbf{d}_{0,\ldots,k-1} \right) + I \left( \mathbf{e}_{k}; \mathbf{d}_{0,\ldots,k-1}, \mathbf{z}_{0} \right). \nonumber
\end{flalign}
Hence,
\begin{flalign} 
2^{ h \left( \mathbf{e}_{k} \right)} 
\geq 2^{h \left( \mathbf{d}_k | \mathbf{d}_{0,\ldots,k-1} \right)}, \nonumber
\end{flalign}
where equality holds if and only if $I \left( \mathbf{e}_{k}; \mathbf{d}_{0,\ldots,k-1}, \mathbf{z}_{0} \right) = 0$.
Therefore,
\begin{flalign}
\left[ \mathbb{E} \left( \left| \mathbf{e}_{k} \right|^{p} \right) \right]^{\frac{1}{p}}
\geq \frac{2^{h \left( \mathbf{d}_k | \mathbf{d}_{0,\ldots,k-1} \right)}}{2 \Gamma \left( \frac{p+1}{p} \right) \left( p \mathrm{e} \right)^{\frac{1}{p}}}, \nonumber
\end{flalign}
where equality holds if and only if $\mathbf{e}_{k}$ is with probability density \eqref{distribution} and $I \left( \mathbf{e}_{k}; \mathbf{d}_{0,\ldots,k-1}, \mathbf{z}_{0} \right) = 0$.
\end{proof}


Note that herein $\mu$ is a normalizing factor. As a matter of fact, when equality is achieved in \eqref{feedbackbound}, it can be verified that
\begin{flalign}
\mu 
= \frac{2^{h \left( \mathbf{d}_k | \mathbf{d}_{0,\ldots,k-1} \right)}}{2 \Gamma \left( \frac{p+1}{p} \right) \left( p \mathrm{e} \right)^{\frac{1}{p}}}.
\end{flalign}
Note also that for the rest of the paper, $\mu$ will always be a normalizing factor as of here, and its value can always be determined in a similar manner as well. Hence, we may skip discussions concerning how to determine $\mu$ for simplicity in the subsequent results.

It is worth mentioning that the only assumptions required for Theorem~\ref{recursivetheorem} to hold is that $g \left( \cdot \right)$ is strictly causal and that $\left\{ \mathbf{d}_{k} \right\}$ and $\mathbf{z}_0$ are independent. Those are in fact very general assumptions, allowing both the controller and the plant to be any causal functions/mappings (as long as the open-loop system is strictly causal), while allowing the disturbance to be with arbitrary distributions. In addition, no conditions on the stationarities of the disturbance or the error signal are required either.

In general, it is seen that the lower bound (for any $p \geq 1$) depends only on the conditional entropy of the current disturbance $\mathbf{d}_{k}$ given the previous disturbances $\mathbf{d}_{0,\ldots,k-1}$, i.e., the amount of ``randomness" contained in $\mathbf{d}_{k}$ given $\mathbf{d}_{0,\ldots,k-1}$. As such, if $\mathbf{d}_{0,\ldots,k-1}$ provide more/less information of $\mathbf{d}_{k}$, then the conditional entropy becomes smaller/larger, and thus the bound becomes smaller/larger. 
In particular, if $\left\{ \mathbf{d}_{k} \right\}$ is a Markov process, then $h \left( \mathbf{d}_k | \mathbf{d}_{0,\ldots,k-1} \right) = h \left( \mathbf{d}_k | \mathbf{d}_{k-1} \right)$ \cite{Cov:06}, and hence \eqref{feedbackbound} reduces to
\begin{flalign} 
\left[ \mathbb{E} \left( \left| \mathbf{e}_{k} \right|^{p} \right) \right]^{\frac{1}{p}}
\geq \frac{2^{h \left( \mathbf{d}_k | \mathbf{d}_{k-1} \right)}}{2 \Gamma \left( \frac{p+1}{p} \right) \left( p \mathrm{e} \right)^{\frac{1}{p}}}.
\end{flalign}
In the worst case, when $\mathbf{d}_{k}$ is independent of $\mathbf{d}_{0,\ldots,k-1}$, we have $h \left( \mathbf{d}_k | \mathbf{d}_{0,\ldots,k-1} \right) = h \left( \mathbf{d}_k \right)$, and thus
\begin{flalign} 
\left[ \mathbb{E} \left( \left| \mathbf{e}_{k} \right|^{p} \right) \right]^{\frac{1}{p}}
\geq \frac{2^{h \left( \mathbf{d}_k \right)}}{2 \Gamma \left( \frac{p+1}{p} \right) \left( p \mathrm{e} \right)^{\frac{1}{p}}}.
\end{flalign}

In addition, equality in \eqref{feedbackbound} holds if and only if the innovation \cite{linearestimation} $\mathbf{e}_{k}$ is with probability \eqref{distribution}, and contains no information of the previous disturbances $\mathbf{d}_{0,\ldots,k-1}$ or initial state $\mathbf{z}_{0}$; it is as if all the ``information" that may be utilized to reduce the $\mathcal{L}_{p}$ norm has been extracted. This is more clearly seen from the viewpoint of entropic innovations, as will be shown shortly. Concerning this, we first present an innovations' perspective to view the term $I \left( \mathbf{e}_{k}; \mathbf{d}_{0,\ldots,k-1}, \mathbf{z}_{0} \right)$. 

\begin{proposition} It holds that
	\begin{flalign} \label{innovation}
	I \left( \mathbf{e}_{k}; \mathbf{d}_{0,\ldots,k-1},  \mathbf{z}_{0} \right) = I \left( \mathbf{e}_{k} ; \mathbf{e}_{0,\ldots,k-1},  \mathbf{z}_{0} \right).
	\end{flalign}
\end{proposition}

\vspace{3mm}

\begin{proof} 
Due to the fact that $\mathbf{z}_{k-1}$ is a function of $\mathbf{d}_{0,\ldots,k-2}$ and $\mathbf{z}_{0}$ (see the proof of Theorem~\ref{recursivetheorem}), we have
\begin{flalign} 
I \left( \mathbf{e}_{k} ; \mathbf{d}_{0,\ldots,k-1}, \mathbf{z}_{0} \right)
& = I \left( \mathbf{e}_{k} ; \mathbf{d}_{0,\ldots,k-2}, \mathbf{e}_{k-1} - \mathbf{z}_{k-1}, \mathbf{z}_{0} \right) \nonumber\\
& = I \left( \mathbf{e}_{k} ; \mathbf{d}_{0,\ldots,k-2}, \mathbf{e}_{k-1}, \mathbf{z}_{0} \right)
.\nonumber
\end{flalign}
Similarly, since  $\mathbf{z}_{i}$ is a function of $\mathbf{d}_{0,\ldots,i-1}$ and $\mathbf{z}_{0}$ for $i = 1, 
\ldots, k- 2$, it follows that 
\begin{flalign} 
& I \left( \mathbf{e}_{k} ; \mathbf{d}_{0,\ldots,k-2}, \mathbf{e}_{k-1}, \mathbf{z}_{0} \right) \nonumber \\
&~~~~ = I \left( \mathbf{d}_{k} ;  \mathbf{d}_{0,\ldots,k-3}, \mathbf{e}_{k-2} - \mathbf{d}_{k-2}, \mathbf{e}_{k-1}, \mathbf{z}_{0} \right) \nonumber \\
&~~~~ = I \left( \mathbf{e}_{k} ;  \mathbf{d}_{0,\ldots,k-3}, \mathbf{e}_{k-2}, \mathbf{e}_{k-1}, \mathbf{z}_{0} \right) \nonumber \\
&~~~~
= \cdots \nonumber \\
&~~~~
= I \left( \mathbf{e}_{k} ;  \mathbf{d}_{0}, \mathbf{e}_{1,\ldots,k-1},  \mathbf{z}_{0} \right)\nonumber \\
&~~~~
= I \left( \mathbf{e}_{k} ;  \mathbf{e}_{0} - \mathbf{z}_{0}, \mathbf{e}_{1,\ldots,k-1}, \mathbf{z}_{0} \right) \nonumber \\
&~~~~
= I \left( \mathbf{e}_{k} ;  \mathbf{e}_{0}, \mathbf{e}_{1,\ldots,k-1}, \mathbf{z}_{0} \right) \nonumber \\
&~~~~
= I \left( \mathbf{e}_{k} ;  \mathbf{e}_{0,\ldots,k-1}, \mathbf{z}_{0} \right), \nonumber
\end{flalign}
which completes the proof. \end{proof}

In addition, if $\mathbf{z}_{0}$ is deterministic (e.g., $\mathbf{z}_{0}=0$), then 
\begin{flalign}
I \left( \mathbf{e}_{k}; \mathbf{d}_{0,\ldots,k-1}, \mathbf{z}_{0} \right) = I \left( \mathbf{e}_{k} ; \mathbf{d}_{0,\ldots,k-1} \right),
\end{flalign}
 and \eqref{innovation} becomes
\begin{flalign}
I \left( \mathbf{e}_{k}; \mathbf{d}_{0,\ldots,k-1} \right) = I \left( \mathbf{e}_{k} ; \mathbf{e}_{0,\ldots,k-1} \right).
\end{flalign}
(Note that more generally,
$
I \left( \mathbf{e}_{k}; \mathbf{e}_{0,\ldots,k-1}, \mathbf{z}_{0} \right) = I \left( \mathbf{e}_{k} ; \mathbf{e}_{0,\ldots,k-1} \right)$
if and only if $h \left(  \mathbf{z}_{0} | \mathbf{e}_{0,\ldots,k} \right) = h \left(  \mathbf{z}_{0} | \mathbf{e}_{0,\ldots,k-1} \right)$.)
In this case, the mutual information between the current innovation and the previous disturbances is equal to that between the current innovation and the previous innovations.
Accordingly, the condition that 
$
I \left( \mathbf{e}_{k}; \mathbf{d}_{0,\ldots,k-1} \right) = 0
$
is equivalent to that 
\begin{flalign}
I \left( \mathbf{e}_{k} ; \mathbf{e}_{0,\ldots,k-1} \right) = 0,
\end{flalign}
which in turn means that the current innovation $\mathbf{e}_{k} $ contains no information of the previous innovations. This is a key link that facilitates the subsequent analysis in the asymptotic case.

\begin{corollary} \label{MIMOasymp}
	Consider the system given in Fig.~\ref{feedback3}. Suppose that the initial state $\mathbf{z}_{0}$ is deterministic.
	Then,  
	\begin{flalign} \label{MIMOasymp1}
	\liminf_{k\to \infty} \left[ \mathbb{E} \left( \left| \mathbf{e}_{k} \right|^{p} \right) \right]^{\frac{1}{p}}
	\geq \liminf_{k\to \infty} \frac{2^{h \left( \mathbf{d}_k | \mathbf{d}_{0,\ldots,k-1} \right)}}{2 \Gamma \left( \frac{p+1}{p} \right) \left( p \mathrm{e} \right)^{\frac{1}{p}}},
	\end{flalign}
	where equality holds if $\left\{ \mathbf{e}_{k} \right\}$ is asymptotically white and with probability density 
	\begin{flalign} \label{asydistribution}
	\lim_{k \to \infty} f_{\mathbf{e}_{k}} \left( x \right)
	= \frac{ \mathrm{e}^{- \left| x \right|^{p} / \left( p \mu^{p} \right)} }{2 \Gamma \left( \frac{p+1}{p} \right) p^{\frac{1}{p}} \mu}.
	\end{flalign}
\end{corollary}

\vspace{3mm}

\begin{proof}
It is known from Theorem~\ref{recursivetheorem} that
\begin{flalign} 
\left[ \mathbb{E} \left( \left| \mathbf{e}_{k} \right|^{p} \right) \right]^{\frac{1}{p}}
\geq \frac{2^{h \left( \mathbf{d}_k | \mathbf{d}_{0,\ldots,k-1} \right)}}{2 \Gamma \left( \frac{p+1}{p} \right) \left( p \mathrm{e} \right)^{\frac{1}{p}}}, \nonumber
\end{flalign} 
where equality holds if and only if $\mathbf{e}_{k}$ is with probability density \eqref{distribution} and $I \left( \mathbf{e}_{k}; \mathbf{d}_{0,\ldots,k-1} \right) = 0$. This, by taking $\liminf_{k\to \infty}$ on its both sides, then leads to
\begin{flalign} 
\liminf_{k\to \infty} \left[ \mathbb{E} \left( \left| \mathbf{e}_{k} \right|^{p} \right) \right]^{\frac{1}{p}}
\geq \liminf_{k\to \infty} \frac{2^{h \left( \mathbf{d}_k | \mathbf{d}_{0,\ldots,k-1} \right)}}{2 \Gamma \left( \frac{p+1}{p} \right) \left( p \mathrm{e} \right)^{\frac{1}{p}}}. \nonumber
\end{flalign}
Herein, equality holds if $\mathbf{e}_{k}$ is with probability density \eqref{distribution} and 
\begin{flalign}
I \left( \mathbf{e}_{k}; \mathbf{d}_{0,\ldots,k-1} \right)  = I \left( \mathbf{e}_{k} ; \mathbf{e}_{0,\ldots,k-1} \right)= 0, \nonumber
\end{flalign}
as $k\to \infty$. Since the fact that
$
I \left( \mathbf{e}_{k} ; \mathbf{e}_{0,\ldots,k-1} \right)= 0
$
as $k\to \infty$ is equivalent to the fact that $\mathbf{e}_{k}$ is asymptotically white, equality in \eqref{MIMOasymp1} holds if $\left\{ \mathbf{e}_{k} \right\}$ is asymptotically white and with probability density \eqref{asydistribution}.
\end{proof}

Strictly speaking, herein ``white" should be ``independent (over time)"; in the rest of the paper, however, we will use ``white"
to replace ``independent" for simplicity, unless otherwise specified. 
On the other hand, when the disturbance is further assumed to be asymptotically stationary, the following corollary holds.

\begin{corollary} \label{uniform}
	Consider the system given in Fig.~\ref{feedback3} with an asymptotically stationary disturbance $\left\{ \mathbf{d}_{k} \right\}$. Suppose that the initial state $\mathbf{z}_{0}$ is deterministic.
	Then,
	\begin{flalign} 
	\liminf_{k\to \infty} \left[ \mathbb{E} \left( \left| \mathbf{e}_{k} \right|^{p} \right) \right]^{\frac{1}{p}}
	\geq \frac{2^{h_{\infty} \left( \mathbf{d} \right)}}{2 \Gamma \left( \frac{p+1}{p} \right) \left( p \mathrm{e} \right)^{\frac{1}{p}}},
	\end{flalign}
	where $h_{\infty} \left( \mathbf{d} \right)$ denotes the entropy rate of $\left\{ \mathbf{d}_{k} \right\}$. Herein, equality holds if $\left\{ \mathbf{e}_{k} \right\}$ is asymptotically white and with probability density \eqref{asydistribution}.
\end{corollary}

\begin{proof}
	Corollary~\ref{uniform} follows directly from Corollary~\ref{MIMOasymp} by noting that 
	\begin{flalign} 
	\liminf_{k\to \infty} h \left( \mathbf{d}_k | \mathbf{d}_{0,\ldots,k-1} \right) = \lim_{k\to \infty} h \left( \mathbf{d}_k | \mathbf{d}_{0,\ldots,k-1} \right)  = h_{\infty} \left( \mathbf{d} \right) \nonumber
	\end{flalign}
	holds for an asymptotically stationary $\left\{ \mathbf{d}_{k} \right\}$ \cite{Cov:06}.
\end{proof}

As a matter of fact, if $\left\{ \mathbf{e}_{k} \right\}$ is asymptotically white and with probability density \eqref{asydistribution}, then, noting also that $\left\{ \mathbf{d}_{k} \right\}$ is asymptotically stationary, it holds that
\begin{flalign} \label{equality}
\lim_{k\to \infty} \left[ \mathbb{E} \left( \left| \mathbf{e}_{k} \right|^{p} \right) \right]^{\frac{1}{p}}
= \frac{2^{h_{\infty} \left( \mathbf{d} \right)}}{2 \Gamma \left( \frac{p+1}{p} \right) \left( p \mathrm{e} \right)^{\frac{1}{p}}}.
\end{flalign}
In addition, we can show that \eqref{equality} holds if and only if $\left\{ \mathbf{e}_{k} \right\}$ is asymptotically white and with probability density \eqref{asydistribution}; in other words, the necessary and sufficient condition for achieving the prediction bounds asymptotically is that the innovation is asymptotically white and with probability density \eqref{asydistribution}. 

\subsection{Special Cases} \label{special}

We now consider the special cases of Theorem~\ref{recursivetheorem} for when $p=2$ and $p=\infty$, respectively.

%

\subsubsection{When $p=2$} \label{variancesection}
The next corollary follows.

\begin{corollary}
	Consider the system given in Fig.~\ref{feedback3}.
	Then,  
	\begin{flalign} \label{variance}
	\left[ \mathbb{E} \left( \mathbf{e}_{k}^{2} \right)  \right]^{\frac{1}{2}}
	\geq \frac{2^{h \left( \mathbf{d}_k | \mathbf{d}_{0,\ldots,k-1} \right)}}{\left( 2 \pi \mathrm{e} \right)^{\frac{1}{2}}},
	\end{flalign}
	where equality holds if and only if $\mathbf{e}_{k}$ is with probability density 
	\begin{flalign}
	f_{\mathbf{e}_{k}} \left( x \right)
	= \frac{ \mathrm{e}^{- x^{2} / \left( 2 \mu^{2} \right)} }{ \left(2 \pi \mu^2 \right)^{\frac{1}{2}} },
	\end{flalign}
	and 
	\begin{flalign}
	\mu = \frac{2^{h \left( \mathbf{d}_k | \mathbf{d}_{0,\ldots,k-1} \right)}}{\left( 2 \pi \mathrm{e} \right)^{\frac{1}{2}}}, 
	\end{flalign}
	that is to say, if and only if $\mathbf{e}_{k}$ is Gaussian,
	and $I \left( \mathbf{e}_{k}; \mathbf{d}_{0,\ldots,k-1}, \mathbf{z}_{0} \right) = 0$.
\end{corollary}

It is clear that \eqref{variance} can simply be rewritten as
\begin{flalign} 
\mathbb{E} \left( \mathbf{e}_{k}^{2} \right)
\geq \frac{2^{2 h \left( \mathbf{d}_k | \mathbf{d}_{0,\ldots,k-1} \right)}}{ 2 \pi \mathrm{e}},
\end{flalign}
which provides a fundamental lower bound for minimum-variance control \cite{aastrom2012introduction}.

\subsubsection{When $p=\infty$} The next corollary follows.

\begin{corollary} 
	Consider the system given in Fig.~\ref{feedback3}.
	Then,  
	\begin{flalign} \label{MD1}
	\esssup_{ f_{\mathbf{e}_{k}} \left( x \right) > 0} \left| \mathbf{e}_{k} \right|
	\geq \frac{2^{h \left( \mathbf{d}_k | \mathbf{d}_{0,\ldots,k-1} \right)}}{2},
	\end{flalign}
	where equality holds if and only if $\mathbf{e}_{k}$ is with probability density 
	\begin{flalign}
	f_{\mathbf{e}_{k}} \left( x \right)
	= \left\{ \begin{array}{cc}
	\frac{1}{2 \mu}, & \left| x \right| \leq \mu,\\
	0, & \left| x \right| > \mu,
	\end{array} \right. 
	\end{flalign}
	and 
	\begin{flalign}
	\mu = \frac{2^{h \left( \mathbf{d}_k | \mathbf{d}_{0,\ldots,k-1} \right)}}{2}, 
	\end{flalign}
	that is to say, if and only if $\mathbf{e}_{k}$ is uniform, and $I \left( \mathbf{e}_{k}; \mathbf{d}_{0,\ldots,k-1}, \mathbf{z}_{0} \right) = 0$.
\end{corollary}


It is worth pointing out that in the case where the variance of the error is minimized (see Section~\ref{variancesection}), it is possible that the probability of having an arbitrary large error (deviation) is non-zero, such is the case when the error is Gaussian \cite{linearestimation}. This could cause severe consequences in safety-critical systems interacting with real world, especially in scenarios where worst-case performance
guarantees must be strictly imposed. Instead, we may directly consider the worst-case scenario by minimizing the maximum (supremum) deviation rather than the variance of the error in the first place. In this case, \eqref{MD1} provides a generic lower bound for the least maximum deviation control; on the other hand, it is also implicated that the error should be made as close to being with a uniform distribution as possible in order to make the maximum deviation as small as possible.

\subsection{Generality of the Bounds} \label{loop}


Note that for the fundamental performance limitations derived in this paper, the classes of control algorithms that can be applied are not restricted as long as they are causal. This means that the performance bounds are valid for all possible control design methods in practical use, including conventional methods as well as machine learning approaches such as reinforcement learning and deep learning (see, e.g., \cite{
	duriez2017machine,
	kiumarsi2017optimal, bertsekas2019reinforcement,recht2019tour,
	zoppoli2020neural} and the references therein); note that any machine learning algorithms that are in the position of the controller  can be viewed as causal functions/mappings from the controller input to the controller output, no matter what the specific algorithms are or how the parameters are to be tuned. As such, the aforementioned fundamental limitations are still valid with any learning elements in the feedback loop; in other words, fundamental limits in general exist to what learning algorithms can achieve in the position of the controller, featuring fundamental limits of learning-based control. (It is true, for instance, that multilayer feedforward neural networks are universal approximators \cite{hornik1989multilayer, goodfellow2016deep}, but it is also true that the performance bounds hold for any functions the neural networks might approximate.)
On the other hand, the classes of plants are not restricted either, as long as they are causal. As such, the fundamental limitations are in fact prevalent in all possible feedback systems.
It is also worth mentioning that no specific restrictions have been imposed on the distributions of the disturbance in general.

\section{Conclusion}

In this paper, we have obtained fundamental performance limitations for generic feedback systems in which both the controller and the plant may be any causal functions/mappings  while the disturbance can be with any distributions. 
Possible future research directions include examining the implications of the bounds in the context of generic state estimation systems, as well as investigating more restricted classes of system dynamics, whereas the current paper focuses on the formulation of the general case.

\addtolength{\textheight}{-12cm}   

\balance

\bibliographystyle{IEEEtran}
\bibliography{references}

\end{document}